\documentclass[11pt,a4paper]{article}
\usepackage[T1]{fontenc}
\usepackage{lmodern}
\usepackage[margin=2.5cm]{geometry}
\usepackage{amsmath,amssymb,amsthm}
\usepackage{microtype}
\usepackage[colorlinks=true,urlcolor=blue!50!black,linkcolor=blue!50!black,citecolor=blue!50!black]{hyperref}
\usepackage{xcolor}
\newtheorem{theorem}{Theorem}[section]

\hypersetup{
  pdftitle={Scalene Yang-Baxter triples realize all nearest-neighbour Hamiltonians with global additive symmetries},
  pdfauthor={Pramod Padmanabhan, Somnath Maity, Akash Sinha, Vladimir Korepin}
}
\begin{document}
\begin{center}
{\LARGE\bfseries
Scalene Yang-Baxter triples realize all\\[3pt]
nearest-neighbour Hamiltonians\\[3pt]
with additive conserved charges\par}
\vspace{1.3em}
{\large
Pramod Padmanabhan\textsuperscript{1},
Somnath Maity\textsuperscript{1},
Akash Sinha\textsuperscript{1}\\[5pt]
and Vladimir Korepin\textsuperscript{2}\par}
\vspace{1em}
{\small
\textsuperscript{1}Department of Physics, School of Basic Sciences,\\
Indian Institute of Technology Bhubaneswar, Bhubaneswar 752050, India\par
\vspace{0.65em}
\textsuperscript{2}C. N. Yang Institute for Theoretical Physics,\\
Stony Brook University, New York 11794, USA\par}
\vspace{1em}
{\footnotesize
\href{mailto:pramod23phys@gmail.com}{\texttt{pramod23phys@gmail.com}}\quad
\href{mailto:somnathmaity126@gmail.com}{\texttt{somnathmaity126@gmail.com}}\\[4pt]
\href{mailto:akash26121999@gmail.com}{\texttt{akash26121999@gmail.com}}\quad
\href{mailto:vladimir.korepin@stonybrook.edu}{\texttt{vladimir.korepin@stonybrook.edu}}\par}
\end{center}
\begin{abstract}
We establish that bond Hamiltonian densities, up to local
difference terms whose sum vanishes under periodic boundary conditions, lying in the commutant of bond operators, $w_j\pm w_{j+1}$, is a necessary and sufficient condition for a one-dimensional nearest-neighbour Hamiltonian to conserve a uniform or staggered sum of the one-site operators $w_j$. Starting from these local commutants, we show that every such Hamiltonian and its
conserved charge can be obtained from monodromy matrices
constructed using regular and non-regular members of scalene
Yang-Baxter triples. The results hold on periodic chains
of length $L\geq3$, with even length required in the staggered
case. The results do not require translation invariance, and hence the interaction on each bond can be chosen independently from the full corresponding commutant. Our representation-independent
construction applies in arbitrary finite local dimension
and exhausts the Hamiltonians with either prescribed additive
symmetry, without requiring ordinary Yang-Baxter integrability.
We illustrate the construction with spin-$\tfrac12$ and spin-1
systems and discuss its relation to the double-commutant
theorem. Classifying the one-site generators $w$ and their
associated bond commutants provides a route to classifying
such Hamiltonians with additive conserved charges.
\end{abstract}

\section{Introduction}
\label{sec:introduction}
Hamiltonians with global symmetries provide a natural setting for exploring collective phenomena in one-dimensional quantum spin chains. Among these symmetries, those generated by additive charges impose direct constraints on the allowed interactions and the resulting dynamics. For Hermitian generators, they prevent dynamical mixing between distinct charge sectors and impose selection rules on matrix elements of observables, restricting the transitions that can contribute to physical response functions~\cite{MoudgalyaMotrunich2023}. Conservation of total magnetization, for example, gives a local continuity equation which then implies that a localized excess of magnetization must redistribute through spin currents, making the conservation law a starting point for understanding spin transport and the relaxation of magnetization profiles~\cite{BertiniTransport}. 

Typically such Hamiltonians may not have too many conservation laws other than the stated additive charges. Apart from some exceptions, they are also not expected to be quantum integrable and more particularly, they may not be obtainable from a solution of the Yang-Baxter equation. Thus it is important to develop a systematic method to construct such Hamiltonians from some first principles. In this work we lay such a foundation, by showing that the {\it scalene Yang-Baxter equation} provides such a framework.

We will consider nearest-neighbour spin-chain Hamiltonians.  Translational invariance is not a requirement for our results. We only work on closed one-dimensional lattices with the Hamiltonians obeying periodic boundary conditions. All the Hamiltonians we consider have additive conserved charges that are either uniform or staggered. The Hamiltonian and the corresponding charges take the forms
\begin{equation}
 H=\sum_{j=1}^{L}h_{j,j+1},\qquad
 Q_{+}=\sum_{j=1}^{L}w_j,\qquad
 Q_{-}=\sum_{j=1}^{L}(-1)^j w_j,
 \label{eq:intro-charges}
\end{equation}
where $w$ is a fixed one-site operator. All indices obey modulo $L$ arithmetic. As declared earlier, the densities $h_{j,j+1}$ may vary independently from bond to bond; translation invariance is not assumed. In particular, different bonds may carry completely different interaction structures. This is stronger than including a site-dependent coupling of one fixed Hamiltonian density to break the translational invariance. We take $L\geq3$ in both cases, with an additional condition of even $L$ when considering $Q_{-}$. When $w$ is Hermitian, these charges generate continuous unitary symmetries. Our construction also allows more general $w$, treating symmetry through the conservation equations $[H,Q_{\pm}]=0$. The problem we address is to realize the Hamiltonians obeying these equations, together with their conserved charges, within a common algebraic framework.

A natural starting point to formulate such a construction is the quantum inverse scattering method. The latter provides a rigorous methodology to construct Hamiltonians and conserved quantities from the ordinary Yang-Baxter equation~\cite{Faddeev,KBI,Slavnov}. In this framework, a regular $R$-matrix gives a commuting family of transfer matrices, whose logarithmic derivatives generate the Hamiltonian and further conserved charges. This construction, however, imposes stronger conditions than conservation of $Q_{+}$ or $Q_{-}$ alone. To encompass {\it all} Hamiltonians with either additive symmetry, we therefore seek a Yang-Baxter-type relation that retains an algebraic connection between the Hamiltonian and its charge without requiring ordinary Yang-Baxter integrability. The scalene Yang-Baxter equation provides such a relation,
\begin{equation}
 A_{12}B_{13}C_{23}=C_{23}B_{13}A_{12},
 \label{eq:intro-scalene}
\end{equation}
whose solutions are ordered triples of generally distinct operators. Relations of this form were studied by Hietarinta and Viallet in their analysis of low-dimensional solutions and their parametrizations~\cite{HietarintaViallet}. An algebraic approach based on multiparameter quantum affine spaces was subsequently developed to construct a large class of scalene solutions~\cite{QuantumAffineSpaces}. An explicit spin-chain construction also demonstrated how a scalene triple can encode conserved operators beyond the ordinary commuting-transfer-matrix framework~\cite{PadmanabhanMaityKorepin}. Here we give a general, representation-independent construction of the scalene triples $(A,B,C)$, that capture both types of additive charges in \eqref{eq:intro-charges}, valid in arbitrary finite local dimension.

Our construction begins with a single site operator $w$ and is organized around the commutants of the local bond operators $w_j\pm w_{j+1}$. Given this we show that there exists at least one scalene triple, such that the braided counterpart of the $B$ operator in the scalene triple satisfies   
\begin{equation*}
 \mathfrak C_{\pm}(w)
 =\bigl\{\check B_{j,j+1}:[\check B_{j,j+1},w_j\pm w_{j+1}]=0\bigr\}.
\end{equation*}
By enforcing regularity for the $B$ operator, the desired Hamiltonian density $h_{j,j+1}$, obtained from the monodromy matrix constructed using $B$, also satisfies the above bond commutant relation. Furthermore, membership of the Hamiltonian density in $\mathfrak C_{\pm}(w)$ is sufficient to show the conservation of $Q_{\pm}$ for the Hamiltonian in \eqref{eq:intro-charges}. Then we show that it is also a necessary condition up to an appropriate change of density,
\begin{equation*}
 h_{j,j+1}\longmapsto h_{j,j+1}+t_j-t_{j+1},
\end{equation*}
where the one-site operators $t_j$ may again depend on $j$. This leaves the periodic Hamiltonian unchanged. Each density in the corresponding local commutant admits a scalene realization. While the Hamiltonian is obtained from the regular member of the triple, the additive charge is encoded in the monodromy constructed from a non-regular member. Thus the construction exhausts the Hamiltonians with the specified symmetry, modulo the difference terms, $t_j-t_{j+1}$, whose sum vanishes under periodic boundary conditions. Henceforth we refer to these difference operators as {\it telescoping terms}. The results obtained in this work hold for all choices of spin representation and is true for both invertible and non-invertible $w$ operators.

Bond and commutant algebras already provide a systematic language for symmetries of local Hamiltonians, including conventional symmetries and Hilbert-space fragmentation~\cite{MoudgalyaMotrunich2022,MoudgalyaMotrunich2023}. The double-commutant theorem also provides a route to constructing operators with prescribed symmetry~\cite{MoudgalyaMotrunich2023}. In this setting our contribution is the explicit scalene Yang-Baxter realization of the full local commutants relevant to uniform and staggered additive charges, together with the global-to-local correspondence modulo telescoping terms. This realizes the Hamiltonian and its symmetry within one ternary relation, without requiring the transfer matrices of the regular member to commute among themselves. For a fixed generator, one pair $A,C$ realizes the entire local commutant, so arbitrary symmetry-preserving changes of the couplings, including spatial disorder, can be incorporated through site-dependent $B$ operators alone. The conserved additive charges organize the Hilbert space into invariant sectors, providing a starting point for a coordinate Bethe ansatz~\cite{Bethe1931,Gaudin}, as illustrated by the fixed-magnetization sectors of the $XXX$ and $XXZ$ chains~\cite{KBI}. In this sense, the scalene construction supplies an algebraic foundation for such an approach; its extension to a complete solution of the spectrum depends on further properties of the interactions.

The paper is organized as follows. Section~\ref{sec:local-scalene} establishes the local sufficiency conditions and presents the representation-independent scalene realizations and the extraction of the additive charges. Section~\ref{sec:necessity} proves the converse from global conservation to a locally commuting density modulo telescoping terms. Section~\ref{sec:examples} gives explicit spin-$\tfrac12$ and spin-1 examples. Section~\ref{sec:discussion-outlook} summarizes the results and discusses double commutants, Haag duality and product symmetries. Section~\ref{sec:outlook} outlines the classification program that can exhaust all such Hamiltonians. Appendix~\ref{app:product-bonds} gives a scalene construction for product-bond commutants.

\section{Local conservation laws and scalene realizations}
\label{sec:local-scalene}
Let $w$ be a one-site operator on a finite-dimensional local Hilbert space $\mathcal{H}$\footnote{While the results obtained here do not depend on the choice of $\mathcal{H}$, we will mostly take this to be the $d$-dimensional complex space $\mathbb{C}^d$, as is standard for quantum spin chains.}. We assume the total Hilbert space is a tensor product of these local Hilbert spaces and so operators with disjoint supports commute. Allow the two-site densities $h_{j,j+1}$ to depend on the bond. Then the local relations
\begin{equation}
 [h_{j,j+1},w_j+w_{j+1}]=0
 \quad\hbox{or}\quad
 [h_{j,j+1},w_j-w_{j+1}]=0
 \label{eq:local-sufficiency}
\end{equation}
imply the following global conservation laws on a periodic chain,
\begin{align}
 [H,Q_+]&=\sum_{j=1}^{L}[h_{j,j+1},w_j+w_{j+1}]=0,
 \label{eq:uniform-sufficiency}\\
 [H,Q_-]&=\sum_{j=1}^{L}(-1)^j[h_{j,j+1},w_j-w_{j+1}]=0.
 \label{eq:staggered-sufficiency}
\end{align}
The second identity requires even $L$. This restriction can be understood directly from the closing bond, $(L,1)$, of an odd periodic chain. Every bulk bond $(j,j+1)$ with $j<L$ contributes $(-1)^j[h_{j,j+1},w_j-w_{j+1}]=0$. For odd $L$, however, sites $L$ and $1$ both enter $Q_-$ with a minus sign, so the closing bond instead gives
\begin{equation}
 [H,Q_-]=-[h_{L,1},w_L+w_1]
 =-2[h_{L,1},w_1],
 \qquad L\ \hbox{odd}.
 \label{eq:odd-boundary-defect}
\end{equation}
The second equality uses $[h_{L,1},w_L-w_1]=0$. The failure of conservation is supported entirely on the closing bond. On an odd periodic chain with difference-bond commuting densities, $Q_-$ is therefore conserved precisely when the closing density also commutes with the bond sum. Equivalently, $h_{L,1}$ must commute with $w_L$ and $w_1$ separately. This additional restriction concerns only the Hamiltonian density on the closing bond. However, if we demand translation invariance as well then this is imposed  on the bulk density.

For an open chain, the closing bond is absent, and the same local difference relation conserves $Q_-$ for both even and odd $L$. There are ways to restore these symmetries for the odd case.  For instance, on an odd-sized ring, conservation can also be restored by choosing a closing density that commutes with $w_L+w_1$, while keeping the difference-bond commuting bulk densities. If there exists an invertible one-site operator $g$ satisfying $gwg^{-1}=-w$, such a closing bond density is supplied by
\begin{equation}
 H_g=\sum_{j=1}^{L-1}h_{j,j+1}+g_L^{-1}h_{L,1}g_L,
 \qquad [H_g,Q_-]=0.
 \label{eq:twisted-Hamiltonian}
\end{equation}
Indeed, the transformed closing bond commutes with $w_L+w_1$, because conjugation by $g_L$ sends this sum to $-w_L+w_1$. The existence of $g$ is an additional condition on $w$. For example when $w=S^z$, a rotation by $\pi$ about the $x$-axis supplies a unitary choice and preserves Hermiticity of the Hamiltonian. This construction generally changes the closing bond and establishes conservation for $H_g$, rather than for the original periodic Hamiltonian $H$. We can obtain such $H_g$'s in the QISM framework by using a twisted auxiliary trace. 

Apart from these special instances neither translation invariance of the densities nor additional assumptions on $w$ are needed. This is important for exhausting all the Hamiltonians with such global additive symmetries. The relations in \eqref{eq:local-sufficiency} establish sufficiency for the additive conservation laws. The converse, modulo telescoping changes of density, is proved in Section~\ref{sec:necessity}.

To realize these Hamiltonians within a Yang-Baxter-type framework, we will derive the local bond-commutant relations \eqref{eq:local-sufficiency} from a local ternary relation \eqref{eq:intro-scalene}. We use the non-braided scalene equation~\eqref{eq:intro-scalene} satisfied by the ordered triple $(A,B,C)$ and choose its middle member to be regular. Then we have 
\begin{equation*}
 B_{j,j+1}(u)=P_{j,j+1}\check B_{j,j+1}(u),\qquad
 \check B_{j,j+1}(0)=I,
\end{equation*}
where $\check B$ is the braided counterpart of $B$. Here $P_{j,j+1}$ exchanges the two local Hilbert spaces and is thus the standard permutation operator acting on finite-dimensional spaces. The corresponding density is
\begin{equation*}
 h_{j,j+1}
 =P_{j,j+1}\left.\frac{\partial B_{j,j+1}(u)}{\partial u}\right|_{u=0}
 =\left.\frac{\partial\check B_{j,j+1}(u)}{\partial u}\right|_{u=0}.
\end{equation*}
Consequently, a differentiable family satisfying
\begin{equation}
 [\check B_{j,j+1}(u),w_j\pm w_{j+1}]=0
 \label{eq:braided-bond-condition}
\end{equation}
produces a density obeying the corresponding relation in \eqref{eq:local-sufficiency}. Conversely, any such density can be embedded into the regular family $\check B_{j,j+1}(u)=I+u h_{j,j+1}$. We will now construct $A$ and $C$ operators of the triple such that the scalene equation is equivalent to \eqref{eq:braided-bond-condition}. These constructions ensure that $\check B_{j,j+1}$ lies in the commutant of $w_j\pm w_{j+1}$. The desired Hamiltonian densities will be obtained from the appropriate monodromy matrices constructed using the scalene triple. We postpone the assignment of auxiliary and physical spaces until that construction; for now, the indices in the scalene equation label only its three tensor factors. Spectral parameters are suppressed in the two theorems, since the equivalences hold separately at every parameter value. Each theorem applies independently on every bond, with the same $A,C$ and an independently chosen $B$. 

\begin{theorem}[Scalene realization of the sum-bond commutant]
\label{thm:sum-bond}
Let $w$ be an arbitrary one-site operator, and let $s$ be a nonzero one-site operator satisfying $s^2=0$. For any fixed $\lambda\in\mathbb C\setminus\{0\}$, define
\begin{equation*}
 A_{12}=I+\lambda w_1s_2,\qquad
 C_{23}=I-\lambda s_2w_3,\qquad
 B_{13}=P_{13}\check B_{13}.
\end{equation*}
Then $A_{12}$ and $C_{23}$ are invertible, and
\begin{equation*}
 A_{12}B_{13}C_{23}=C_{23}B_{13}A_{12}
 \quad\Longleftrightarrow\quad
 [\check B_{13},w_1+w_3]=0.
\end{equation*}
For every allowed choice of $s$ and $\lambda$, the possible braided middle members are exactly the sum-bond commutant. The realization of a given $B$ by $A$ and $C$ need not be unique.
\end{theorem}

\begin{proof}
Using $s_2^2=0$, the inverses are $A_{12}^{-1}=I-\lambda w_1s_2$ and $C_{23}^{-1}=I+\lambda s_2w_3$. Substituting the algebraic expressions for $A$, $B$ and $C$ into \eqref{eq:intro-scalene} and moving the permutation to the left gives
\begin{align*}
 A_{12}B_{13}C_{23}
 &=P_{13}(I+\lambda w_3s_2)\check B_{13}(I-\lambda s_2w_3),\\
 C_{23}B_{13}A_{12}
 &=P_{13}(I-\lambda s_2w_1)\check B_{13}(I+\lambda w_1s_2).
\end{align*}
Since $s_2$ commutes with every operator on spaces $1$ and $3$, subtracting these expressions yields
\begin{equation*}
 A_{12}B_{13}C_{23}-C_{23}B_{13}A_{12}
 =\lambda P_{13}s_2[w_1+w_3,\check B_{13}].
\end{equation*}
The permutation is invertible and $\lambda\ne0$. Moreover, $s_2\ne0$ acts on a separate factor, so the right-hand side vanishes precisely when the commutator on factors $1$ and $3$ vanishes. 
\end{proof}

For every $h_{j,j+1}\in\mathfrak C_+(w)$, taking $\check B_{j,j+1}(u)=I+u h_{j,j+1}$ therefore gives a regular scalene middle member with the fixed $A$ and $C$ of Theorem~\ref{thm:sum-bond}. Its density conserves $Q_+$ by \eqref{eq:uniform-sufficiency}. We now formulate the scalene triple for the difference bond operator $w_j - w_{j+1}$. In this case we introduce another one-site grading operator $r$ together with the nilpotent operator $s$. Their anticommutation supplies the relative sign between the two sites in the bond.

\begin{theorem}[Scalene realization of the difference-bond commutant]
\label{thm:difference-bond}
Let $w$ be an arbitrary one-site operator, and let one-site operators $r,s$ obey
\begin{equation}
 r^2=I,\qquad s^2=0,\qquad rs=-sr,\qquad s\ne0.
 \label{eq:auxiliary-grading}
\end{equation}
For any fixed $\lambda\in\mathbb C\setminus\{0\}$, define
\begin{equation*}
 A_{12}=r_2+\lambda w_1s_2,\qquad
 C_{23}=r_2+\lambda s_2w_3,\qquad
 B_{13}=P_{13}\check B_{13}.
\end{equation*}
Then $A_{12}^2=C_{23}^2=I$, and
\begin{equation*}
 A_{12}B_{13}C_{23}=C_{23}B_{13}A_{12}
 \quad\Longleftrightarrow\quad
 [\check B_{13},w_1-w_3]=0.
\end{equation*}
For every allowed choice of $r,s$ and $\lambda$, the possible braided middle members lie exactly in the difference-bond commutant. The realization of a given $B$ by $A$ and $C$ need not be unique.
\end{theorem}

\begin{proof}
The identities $A_{12}^2=C_{23}^2=I$ follow immediately from \eqref{eq:auxiliary-grading}. The permutation exchanges factors $1$ and $3$, so
\begin{equation*}
 A_{12}B_{13}C_{23}-C_{23}B_{13}A_{12}
 =P_{13}\bigl(C_{23}\check B_{13}C_{23}
             -A_{12}\check B_{13}A_{12}\bigr).
\end{equation*}
For $k=1,3$, the relations on $r,s$ give
\begin{equation*}
 (r_2+\lambda s_2w_k)\check B_{13}(r_2+\lambda s_2w_k)
 =\check B_{13}+\lambda (rs)_2[\check B_{13},w_k].
\end{equation*}
Subtracting the two instances of this identity yields
\begin{equation*}
 A_{12}B_{13}C_{23}-C_{23}B_{13}A_{12}
 =P_{13}\lambda (rs)_2[w_1-w_3,\check B_{13}].
\end{equation*}
Since $r$ is invertible and $s\ne0$, one has $rs\ne0$. Since $(rs)_2$ acts on factor $2$ and the commutator acts on factors $1$ and $3$, the tensor-product structure and $\lambda\ne0$ again proves the desired commutation relation.
\end{proof}

Every $h_{j,j+1}\in\mathfrak C_-(w)$ is therefore realized by $\check B_{j,j+1}(u)=I+u h_{j,j+1}$ with the fixed $A$ and $C$ of Theorem~\ref{thm:difference-bond}. The resulting periodic Hamiltonian conserves $Q_-$ for even $L$.

Both proofs use only the displayed algebraic relations and impose no polynomial identity, diagonalizability or Hermiticity condition on $w$. The required one-site operators, $r$ and $s$, exist in every dimension $d\geq2$: in terms of generators of $\rm Mat(\mathbb{C},d)$,  $E_{mn}=|m\rangle\langle n|$, we can take $s=E_{12}$ and $r=I-2E_{22}$. The triples can therefore be realized on three copies of $\mathcal{H}$ in any nontrivial finite local dimension. The case $d=1$ is scalar and satisfies the scalene equation trivially.

The freedom in $A$ and $C$ is distinct from the freedom in the Hamiltonian density. In both theorems, $\lambda$ is a fixed nonzero scalar, independent of the spectral parameter $u$, and can be absorbed into $s$. Retaining it explicitly displays a normalization freedom in the charge-carrying part of the triple. At $\lambda=0$, the scalene relation imposes no bond-commutant condition on $B$, and the non-regular monodromy loses its charge term. We therefore exclude this value when claiming equivalence with the full local commutant.

We can try other ans\"atze for $A$ and $C$ instead of the ones shown in these Theorems. A simple way to do this is by changing the scalar factors into the independent coefficients $\lambda_A$ and $\lambda_C$. Within the normalized ansatz of Theorem~\ref{thm:sum-bond}, take $A_{12}=I+\lambda_Aw_1s_2$ and $C_{23}=I-\lambda_Cs_2w_3$. At the regular point $B_{13}(0)=P_{13}$, direct multiplication gives
\begin{equation}
 A_{12}B_{13}(0)C_{23}-C_{23}B_{13}(0)A_{12}
 =(\lambda_A-\lambda_C)P_{13}s_2(w_3-w_1).
 \label{eq:unequal-sum-normalizations}
\end{equation}
Likewise, for $A_{12}=r_2+\lambda_Aw_1s_2$ and $C_{23}=r_2+\lambda_Cs_2w_3$, the difference construction gives
\begin{equation}
 A_{12}B_{13}(0)C_{23}-C_{23}B_{13}(0)A_{12}
 =(\lambda_C-\lambda_A)P_{13}(rs)_2(w_1+w_3).
 \label{eq:unequal-difference-normalizations}
\end{equation}
For non-scalar $w$, neither $w_3-w_1$ nor $w_1+w_3$ is zero, so the separate tensor factors force $\lambda_A=\lambda_C$ in both cases. This necessity already follows at the regular point and hence applies to every regular family used here, including all the spin examples. For completeness, the sum ansatz permits unequal coefficients if $w$ is scalar; the difference ansatz permits them only if $w=0$. These exceptional choices give only trivial additive charges. The conclusion concerns these fixed ansatz forms; it is not a uniqueness theorem for all scalene solutions with a given $B$.

There are still other choices available for $A$ and $C$. Any admissible $s$ in the sum construction, or pair $r,s$ in the difference construction, works with the same $B$. Moreover, for any invertible one-site operator $Y$, the auxiliary-space change of basis
\begin{equation*}
 A_{12}\longmapsto Y_2A_{12}Y_2^{-1},\qquad
 C_{23}\longmapsto Y_2C_{23}Y_2^{-1}
\end{equation*}
preserves the scalene relation while leaving $B_{13}$ unchanged, because $Y_2$ commutes with $B_{13}$. Independent overall rescalings $A\mapsto c_AA$ and $C\mapsto c_CC$, with nonzero scalars $c_A,c_C$, also preserve the relation: both sides acquire the same factor $c_Ac_C$. Such overall factors need not agree, unlike the relative nilpotent coefficients in \eqref{eq:unequal-sum-normalizations} and \eqref{eq:unequal-difference-normalizations}. These freedoms need not give inequivalent triples, and a particular $B$ can admit additional realizations outside our ansatz. Our exhaustion statement concerns Hamiltonian densities modulo telescoping terms, not a classification of all triples realizing them.

\subsection*{Monodromy matrices and global conservation}
We will now adapt QISM~\cite{Faddeev,KBI,Slavnov} to the scalene triples constructed earlier. In the local scalene equation, let $1=a$ and $2=b$ be auxiliary spaces and let $3=j$ be a physical site. Thus $A_{ab}$ intertwines the two auxiliary spaces, while $B_{aj}(u)$ and $C_{bj}$ couple them to the physical sites. We define the monodromy matrices corresponding to the $B$ and $C$ operators as
\begin{equation*}
 T^{(B)}_a(u)=B_{aL}(u)\cdots B_{a1}(u),\qquad
 T^{(C)}_b=C_{bL}\cdots C_{b1}.
\end{equation*}
Then the corresponding transfer matrices, obtained by tracing over the auxiliary spaces, are defined as $\tau_B(u)=\operatorname{tr}_a T^{(B)}_a(u)$ and $\tau_C=\operatorname{tr}_b T^{(C)}_b$. In each case, the $\check B_{aj}(u)$ operators are taken independently from the corresponding bond commutant, i.e. $B_{aj}(u)$ need not be translated copies of a single operator family. This freedom is not restricted to scalar rescalings of one interaction or to spectral shifts $B_{aj}(u)=B(u-\theta_j)$ of a common family. Instead, distinct operator families can all be regular at the same point $u=0$. The same $A_{ab}$ and $C_{bj}$ are used throughout the chain. Repeated application of the local scalene relation gives the $RTT$-like identity
\begin{equation}
 A_{ab}T^{(B)}_a(u)T^{(C)}_b
 =T^{(C)}_bT^{(B)}_a(u)A_{ab}.
 \label{eq:mixed-RTT}
\end{equation}
To see this, commute factors on disjoint auxiliary and physical sites to arrange the product as $(B_{aL}C_{bL})\cdots(B_{a1}C_{b1})$, and move $A_{ab}$ through these pairs one site at a time. This argument assumes no ordinary Yang-Baxter equation for $B$. 

We will now simplify the $C$-monodromy matrices of our scalene triples and extract the additive conserved charges. For the sum-bond construction, $C_{bj}=I-\lambda s_bw_j$. Nilpotency makes every term containing two or more factors of $s_b$ vanish, giving a compact expression for the associated monodromy matrix on an arbitrary number of sites $L$,
\begin{equation}
 T^{(C)}_b
 =\prod_{j=L}^{1}(I-\lambda s_bw_j)
 =I-\lambda s_b\sum_{j=1}^{L}w_j
 =I-\lambda s_bQ_+.
 \label{eq:uniform-monodromy}
\end{equation}
The uniform charge is therefore the coefficient of $-\lambda s_b$ in the monodromy. Since $s_b$ is nilpotent, $\operatorname{tr}_b s_b=0$, and hence
\begin{equation*}
 \tau_C=d_b I,\qquad d_b=\dim\mathcal{H}_b.
\end{equation*}
Thus taking the auxiliary trace removes the charge from this monodromy and produces a trivial transfer matrix.

A similar result holds for the difference-bond construction. In this case we have $C_{bj}=r_b+\lambda s_bw_j$. Using $r_b^2=I$ and $r_bs_b=-s_br_b$, rewrite each factor as $r_b(I+\lambda (rs)_bw_j)$. Moving all powers of $r_b$ to the left gives
\begin{align}
 T^{(C)}_b
 &=r_b^L\prod_{j=L}^{1}\bigl(I+(-1)^{j-1}\lambda (rs)_bw_j\bigr)\notag\\
 &=r_b^L\bigl(I-\lambda (rs)_bQ_-\bigr).
 \label{eq:staggered-monodromy-general}
\end{align}
Here $(rs)_b^2=0$, so once again only terms linear in the charge survive. On an even periodic chain this becomes
\begin{equation}
 T^{(C)}_b=I-\lambda (rs)_bQ_-,\qquad \tau_C=d_b I,
 \label{eq:staggered-monodromy}
\end{equation}
where the trace identity follows from the nilpotency of $rs$. The staggered charge is encoded in the coefficient of $-\lambda(rs)_b$, although its transfer matrix is proportional to the identity. The corresponding odd-length $RTT$ calculation is given below. In both cases the relevant coefficient matrix is nonzero and linearly independent of $I$, so the charge can be recovered directly from the monodromy after a trace on the auxiliary matrix space.

These scalar transfer matrices make the traced cross-commutativity relation $[\tau_B(u),\tau_C]=0$ uninformative. The untraced relation~\eqref{eq:mixed-RTT}, however, retains the additive charge and yields its conservation. For the uniform construction, substitute $A_{ab}=I+\lambda w_as_b$ and \eqref{eq:uniform-monodromy} into \eqref{eq:mixed-RTT}. Expanding and using $s_b^2=0$ gives
\begin{equation*}
 \lambda s_b\bigl([w_a,T^{(B)}_a(u)]+[Q_+,T^{(B)}_a(u)]\bigr)=0.
\end{equation*}
The expression in parentheses acts independently of space $b$. Since $\lambda\ne0$ and $s_b\ne0$, this proves
\begin{equation}
 [T^{(B)}_a(u),w_a+Q_+]=0,
 \qquad [\tau_B(u),Q_+]=0.
 \label{eq:uniform-RTT-charge}
\end{equation}
The second identity follows by tracing over $a$, for which the trace of the commutator with $w_a$ vanishes. This computation is independent of the auxiliary space $b$.

For the staggered construction on an even chain, substitute $A_{ab}=r_b+\lambda w_as_b$ and \eqref{eq:staggered-monodromy} instead. The two sides of \eqref{eq:mixed-RTT} reduce to
\begin{align*}
 A_{ab}T^{(B)}_aT^{(C)}_b
 &=r_bT^{(B)}_a+\lambda s_b\bigl(w_aT^{(B)}_a-T^{(B)}_aQ_-\bigr),\\
 T^{(C)}_bT^{(B)}_aA_{ab}
 &=r_bT^{(B)}_a+\lambda s_b\bigl(T^{(B)}_aw_a+Q_-T^{(B)}_a\bigr).
\end{align*}
Since $\lambda\ne0$ and $s_b\ne0$ acts on a separate factor, we obtain
\begin{equation}
 [w_a,T^{(B)}_a(u)]=\{T^{(B)}_a(u),Q_-\},
 \qquad \{\tau_B(u),Q_-\}=0,
 \label{eq:staggered-RTT-charge}
\end{equation}
where the second equality follows again from a trace over the auxiliary space $a$. It is natural to expect an anticommutator in this case as the regular transfer matrix reduces to the one-site translation operator at $u=0$, which reverses the staggered charge.

To extract the Hamiltonian, use regularity $B_{aj}(0)=P_{aj}$ at every site. Then $\tau_B(0)=U$, the one-site cyclic shift on the closed chain. To keep track of bond dependence, write $\check B_{aj}(u)=P_{aj}B_{aj}(u)$ and let $\eta^{(j)}$ denote its derivative at zero as an operator on $\mathcal H\otimes\mathcal H$. The permutation calculation gives
\begin{equation}
 H:=U^{-1}\tau'_B(0)=\sum_{j=1}^L\eta^{(j)}_{j-1,j}.
 \label{eq:inhomogeneous-derivative}
\end{equation}
Here $\eta^{(j)}_{j-1,j}$ denotes the copy of $\eta^{(j)}$ on the ordered physical bond $(j-1,j)$. Thus a prescribed density $h_{j,j+1}$ is assigned to the factor $B_{a,j+1}$: choose $\eta^{(j+1)}_{j,j+1}=h_{j,j+1}$, with cyclic indices. The linear regular families constructed above realize any such assignment. This identity uses regularity alone, as in the standard permutation argument~\cite{Faddeev,KBI,Slavnov}; it requires neither translation invariance nor self-commutativity of the transfer family. In particular, $\tau_B(0)=U$ does not imply $[H,U]=0$.

In the uniform case, \eqref{eq:uniform-RTT-charge} implies that both $U$ and $\tau'_B(0)$ commute with $Q_+$, so $[H,Q_+]=0$. In the staggered case, \eqref{eq:staggered-RTT-charge} implies
\begin{equation*}
 UQ_-=-Q_-U,\qquad \tau'_B(0)Q_-=-Q_-\tau'_B(0).
\end{equation*}
The two minus signs cancel in the logarithmic derivative:
\begin{equation*}
 HQ_-=U^{-1}\tau'_B(0)Q_-
 =-U^{-1}Q_-\tau'_B(0)
 =Q_-U^{-1}\tau'_B(0)=Q_-H.
\end{equation*}
Thus the $RTT$-like relation recovers $[H,Q_{\pm}]=0$ in both cases, even though the corresponding $C$-transfer matrices are scalar. The nontrivial symmetry information resides in the $C$-monodromies themselves.

For odd $L$, the boundary obstruction discussed earlier also appears in the $RTT$ calculation. The monodromy identity~\eqref{eq:staggered-monodromy-general} remains valid for odd $L$, but now $r_b^L=r_b$. Hence
\begin{equation}
 T^{(C)}_b=r_b-\lambda s_bQ_-,\qquad
 \tau_C=(\operatorname{tr}_b r_b)I.
 \label{eq:odd-staggered-monodromy}
\end{equation}
The alternating operator $Q_-$ is still present in the monodromy, and the transfer matrix remains scalar. However, for a traceless $r$ it vanishes. Its occurrence as a monodromy coefficient, however, does not by itself establish conservation on an odd periodic chain. Indeed, inserting \eqref{eq:odd-staggered-monodromy} into the same $RTT$-like relation gives
\begin{align*}
 A_{ab}T^{(B)}_aT^{(C)}_b
 &=T^{(B)}_a-\lambda (rs)_b\bigl(w_aT^{(B)}_a+T^{(B)}_aQ_-\bigr),\\
 T^{(C)}_bT^{(B)}_aA_{ab}
 &=T^{(B)}_a+\lambda (rs)_b\bigl(T^{(B)}_aw_a+Q_-T^{(B)}_a\bigr).
\end{align*}
Since $\lambda\ne0$ and $rs\ne0$, their equality implies
\begin{equation}
 \{T^{(B)}_a(u),w_a+Q_-\}=0,
 \qquad
 \{\tau_B(u),Q_-\}=-2\operatorname{tr}_a\bigl(w_aT^{(B)}_a(u)\bigr).
 \label{eq:odd-RTT-charge}
\end{equation}
Thus the traced identity contains an insertion of $w_a$, in place of the anticommutation relation used for even $L$. At the regular point, the one-site shift obeys
\begin{equation*}
 UQ_-U^{-1}=-Q_- - 2w_1,
\end{equation*}
for our monodromy ordering. The extra term records the mismatch of alternating signs on the last bond.

The boundary-twisted Hamiltonian~\eqref{eq:twisted-Hamiltonian} has a corresponding monodromy construction. For the same operator $g$, define
\begin{equation*}
 \tau_{B,g}(u)=\operatorname{tr}_a\bigl(g_aT^{(B)}_a(u)\bigr).
\end{equation*}
Multiplying the odd-length $RTT$ identity~\eqref{eq:odd-RTT-charge} by $g_a$ and tracing over $a$ gives
\begin{equation*}
 \{\tau_{B,g}(u),Q_-\}
 =-\operatorname{tr}_a\bigl(\{g_a,w_a\}T^{(B)}_a(u)\bigr)=0.
\end{equation*}
Since $\tau_{B,g}(0)=g_1U$ is invertible, the two anticommutations cancel in the logarithmic derivative, yielding
\begin{equation*}
 \tau_{B,g}(0)^{-1}\tau'_{B,g}(0)=H_g,
 \qquad [H_g,Q_-]=0.
\end{equation*}
This twisted trace produces the modified closing bond specified above. A weighted trace of $T^{(C)}_b$ may instead extract $Q_-$, but extraction alone cannot remove the obstruction~\eqref{eq:odd-boundary-defect} for the original periodic Hamiltonian.

\paragraph{Related monodromy constructions :}
Symmetry generators also occur in untraced monodromy matrices in established QISM constructions. For the periodic $XXX$ chain, Faddeev explicitly identifies the total-spin operators in the large-spectral-parameter expansion of the monodromy~\cite[(45) and (98)]{Faddeev}. With the normalization $L_{aj}(u)=uI+i\sum_{\alpha=x,y,z}\sigma_a^\alpha S_j^\alpha$, this reads
\begin{equation*}
 T_a(u)=u^L I+i u^{L-1}\sum_{\alpha=x,y,z}\sigma_a^\alpha S_{\mathrm{tot}}^\alpha
 +O(u^{L-2}),\qquad
 S_{\mathrm{tot}}^\alpha=\sum_{j=1}^{L}S_j^\alpha.
\end{equation*}
The displayed spin-dependent term vanishes under the auxiliary trace because the Pauli matrices are traceless. Nevertheless, the $RTT$ relation implies $[T_a(u),S_{\mathrm{tot}}^\alpha+\sigma_a^\alpha/2]=0$, and hence conservation of $S_{\mathrm{tot}}^\alpha$ by the periodic Hamiltonian. Thus even in this standard example, the individual spin generators are read from monodromy coefficients rather than from the corresponding coefficient of the transfer matrix.

Another example is the one-dimensional Hubbard model. Uglov and Korepin established its Yangian symmetry on the infinite chain~\cite{UglovKorepin}. Murakami and G\"ohmann subsequently obtained the Yangian generators from the asymptotic expansion of a submatrix of the monodromy in an infinite-interval, zero-density QISM construction~\cite[Sec.~3]{MurakamiGohmann}; the quantum determinant of that submatrix generates Yangian-invariant commuting operators, including the Hamiltonian. The infinite-interval setting is essential to this comparison and should not be identified with a generic finite periodic Yangian symmetry. These examples place the extraction of symmetry generators from monodromy entries in an established context. The construction here provides an explicit representation-independent scalene realization for every Hamiltonian in the specified additive-symmetry classes, with $Q_\pm$ retained in a non-regular monodromy whose auxiliary trace is scalar.

\paragraph{Comparison with the Reshetikhin condition :}
For a homogeneous density obtained from an analytic, regular difference-form solution of the ordinary Yang-Baxter equation, the density $h=\check R'(0)$, with $\check R(0)=I$, necessarily satisfies~\cite{KulishSklyanin1982, JimboMiwa1984, GrabowskiMathieu1995, Bibikov2000}
\begin{equation}
 \mathcal D(h):=[h_{12}+h_{23},[h_{12},h_{23}]]
 =K_{23}-K_{12}
 \label{eq:reshetikhin-condition}
\end{equation}
for some two-site operator $K$. This cubic constraint arises at third order in the spectral expansion. Our scalene equation instead reduces exactly to the linear bond-commutant condition and imposes no such constraint. An explicit algebraic counterexample is the spin-1 density
\begin{equation*}
 h_{j,j+1}=S_j^xS_{j+1}^x+S_j^yS_{j+1}^y
 =\tfrac12(S_j^+S_{j+1}^-+S_j^-S_{j+1}^+).
\end{equation*}
It commutes with $S_j^z+S_{j+1}^z$ and therefore has the scalene realization of Theorem~\ref{thm:sum-bond}. Nevertheless, in the three-site product basis of $S^z$ eigenstates, direct multiplication gives
\begin{equation*}
 \langle1,1,-1|\mathcal D(h)|0,0,1\rangle=1,
 \qquad
 \langle1,1,-1|(K_{23}-K_{12})|0,0,1\rangle=0
 \quad\text{for every }K.
\end{equation*}
The second equality follows because the states differ on both outer sites. Thus \eqref{eq:reshetikhin-condition} fails despite the local symmetry. The obstruction concerns this density representative; it neither tests all telescoping representatives nor excludes more general spectral dependence. Conversely, integrable members of our families can satisfy the condition. The scalene construction therefore accommodates both possibilities without imposing the ordinary Yang-Baxter constraint.

We close this section with a remark on the ordered nature of the scalene triple. The roles of its regular and non-regular members can be exchanged by a corresponding transformation of the intertwiner. Since $A$ is invertible, interchanging spaces $1$ and $2$ in the scalene equation and multiplying by $A_{21}^{-1}$ on both sides gives another ordered solution,
\begin{equation*}
 \widetilde A_{12}=A_{21}^{-1},\qquad
 \widetilde B_{13}=C_{13},\qquad
 \widetilde C_{23}(u)=B_{23}(u),
\end{equation*}
where $A_{21}=P_{12}A_{12}P_{12}$. Thus exchanging $B$ and $C$ requires changing $A$ as well. For the two constructions above, $A_{21}^{-1}=C_{12}$, so the transformed triples take the simple form
\begin{equation*}
 \widetilde A_{12}=M_{12},\qquad
 \widetilde B_{13}=M_{13},\qquad
 \widetilde C_{23}(u)=P_{23}\check B_{23}(u),
 \end{equation*}
 where
 \begin{equation*}
 M_{12}=\begin{cases}
 I-\lambda s_1w_2,&\text{sum bond},\\
 r_1+\lambda s_1w_2,&\text{difference bond}.
 \end{cases}
\end{equation*}
Now $\widetilde C(0)=P$ is regular and produces the same Hamiltonian density, $P\widetilde C'(0)=\check B'(0)=h$. The non-regular member $\widetilde B$ has precisely the monodromy previously computed for $C$, with the auxiliary label changed from $b$ to $a$, and therefore contains $Q_+$ or $Q_-$. The Hamiltonian is consequently generated by $\tau_{\widetilde C}$, while the additive operator is retained in $T^{(\widetilde B)}$. All chain-length and boundary qualifications established above continue to apply.


\section{From global symmetry to local bond commutants}
\label{sec:necessity}
The constructions of Section~\ref{sec:local-scalene} realize every Hamiltonian density in the sum- or difference-bond commutant. To show that they exhaust all nearest-neighbour Hamiltonians, including those with spatially varying interactions, with the corresponding global conserved charge, we must prove the converse statement. The latter states that a nearest-neighbor Hamiltonian with either of the global conservation laws is made up of Hamiltonian densities that lies in the appropriate bond commutant. Such Hamiltonian densities are only defined up to local difference terms $t_j-t_{j+1}$, whose sum vanishes over a periodic chain. For a generic Hamiltonian density, the converse need not hold. In fact the obstruction to their validities are precisely attributed to the local densities of the form $t_j-t_{j+1}$. A simple example illustrates this obstruction. Consider the density
$h_{j,j+1}=\sigma_j^x-\sigma_{j+1}^x$ and take
$w_j=\sigma_j^z$. Its periodic sum reduces the Hamiltonian to 0,  $H=0$,
so both global charges $Q_\pm$ are conserved. Nevertheless,
the local commutators are
\[
 [h_{j,j+1},w_j+w_{j+1}]
 =-2i(\sigma_j^y-\sigma_{j+1}^y),\qquad
 [h_{j,j+1},w_j-w_{j+1}]
 =-2i(\sigma_j^y+\sigma_{j+1}^y),
\]
which are nonzero. More generally, adding such local difference
terms to any locally commuting density leaves the periodic
Hamiltonian unchanged but can destroy its local commutant
relation. Thus global conservation can imply the local relation
only after removing an appropriate telescoping change of density.
We will now prove that this telescoping freedom is the only obstruction.

We work on the full tensor-product Hilbert space of a periodic chain with local dimension $d<\infty$. The densities may depend arbitrarily on the bond; in particular, we impose no reflection symmetry, or any other integrability condition. We take $L\geq3$ for the uniform charge and even $L\geq4$ for the staggered charge. The conservation law is taken to be an operator identity on the full Hilbert space, rather than a relation restricted to a selected charge sector.

We know that every two-site operator admits a unique decomposition in this setup. In particular, a generic Hamiltonian density is given by
\begin{equation}
 h_{j,j+1}=c_jI+\alpha_j+\beta_{j+1}+\kappa_{j,j+1},\qquad
 \operatorname{tr}\alpha_j=\operatorname{tr}\beta_j=0,\qquad
 \operatorname{tr}_{j}\kappa_{j,j+1}
 =\operatorname{tr}_{j+1}\kappa_{j,j+1}=0.
 \label{eq:density-decomposition}
\end{equation}
Here $c_j$ is a site-dependent scalar and $\kappa_{j,j+1}$ contains the part with traceless factors on both sites. The one-site operator $\alpha_j$ is the left contribution of bond $(j,j+1)$, whereas $\beta_j$ is the right contribution of bond $(j-1,j)$. Both act on site $j$, and neither is assumed to be a translate of a fixed operator. It is important to note that this is a decomposition of an arbitrary density, and not an ansatz. More explicitly we have,
\begin{equation*}
 c_j=\frac{\operatorname{tr}h_{j,j+1}}{d^2},\qquad
 \alpha_j=\frac{1}{d}\operatorname{tr}_{j+1}h_{j,j+1}-c_jI_j,\qquad
 \beta_{j+1}=\frac{1}{d}\operatorname{tr}_{j}h_{j,j+1}-c_jI_{j+1},
\end{equation*}
and $\kappa$ is the remainder in \eqref{eq:density-decomposition}. By periodicity, the Hamiltonian corresponding to \eqref{eq:density-decomposition} becomes
\begin{equation}
 H=L\bar cI+\sum_{j=1}^{L}v_j+\sum_{j=1}^{L}\kappa_{j,j+1},
 \qquad v_j=\alpha_j+\beta_j,\quad \bar c=L^{-1}\sum_jc_j.
 \label{eq:canonical-global-density}
\end{equation}
As an illustration, let $\mathcal{H}=\mathbb{C}^2$. Then an example of a spin-$\tfrac12$ density is given by,
\[
 h_{j,j+1}=p_j\sigma_j^z+q_j\sigma_{j+1}^z
 +J_{x,j}\sigma_j^x\sigma_{j+1}^x
 +J_{y,j}\sigma_j^y\sigma_{j+1}^y
 +J_{z,j}\sigma_j^z\sigma_{j+1}^z.
\]
Here $p_j,q_j,J_{x,j},J_{y,j},J_{z,j}$ are arbitrary complex, bond-dependent coefficients. Since the Pauli matrices are traceless, $c_j=0$, $\alpha_j=p_j\sigma_j^z$, $\beta_{j+1}=q_j\sigma_{j+1}^z$, and $\kappa_{j,j+1}$ consists of the three exchange terms. Each exchange term has zero partial trace over either site. In the periodic sum the field terms combine into $\sum_j(p_j+q_{j-1})\sigma_j^z$: the physical field $v_j$ receives contributions from the two adjacent bonds. The operators appearing in this Hamiltonian density are linearly independent. This independence can be seen directly from a familiar operator basis. For two spins of local dimension two, the sixteen Pauli products
\[
 I,\qquad \sigma_j^\mu,\qquad \sigma_{j+1}^\nu,\qquad
 \sigma_j^\mu\sigma_{j+1}^\nu,
 \qquad \mu,\nu\in\{x,y,z\},
\]
form a basis of the two-site operator space. There is one identity, three one-site operators on each site, and nine products with a traceless factor on both sites. These groups are mutually orthogonal under the Hilbert-Schmidt inner product, so a linear combination of the field terms cannot cancel a linear combination of the interaction terms. The example above uses a subset of this basis, whereas an arbitrary two-spin density may contain all sixteen components.

The same reasoning holds in an arbitrary local dimension $d$. Choose traceless one-site operators $F^\mu$, $\mu=1,\ldots,d^2-1$, such that $I,F^1,\ldots,F^{d^2-1}$ form a basis of $\operatorname{End}(\mathbb{C}^d)$. Then
\[
 I,\qquad F_j^\mu,\qquad F_{j+1}^\nu,\qquad
 F_j^\mu F_{j+1}^\nu,
 \qquad \mu,\nu=1,\ldots,d^2-1,
\]
form a basis of $\operatorname{End}(\mathcal{H}_j\otimes\mathcal{H}_{j+1})$, with
$1+2(d^2-1)+(d^2-1)^2=d^4$ elements. The four groups span precisely the scalar, left one-site, right one-site and two-site components with vanishing partial traces in \eqref{eq:density-decomposition}. Including the identity is essential for a basis of the entire two-site operator space. This also makes clear that the interaction component $\kappa$ excludes terms that merely act as the identity on one of its sites.

On the full chain, products of these one-site basis elements form a basis of the global operator space. Traceless one-site terms have a nonidentity factor at one site, whereas the terms in $\kappa_{j,j+1}$ have nonidentity factors at both sites of their bond. They therefore involve different basis elements and cannot cancel. For $L\geq3$, distinct nearest-neighbour bonds also have distinct sets of two sites, so their zero-partial-trace contributions are independent even when the bonds share a site. Thus the identity, one-site, and nearest-neighbour two-site terms span the subspace relevant to \eqref{eq:canonical-global-density}. In physical terms they correspond to the constant energy shift, magnetic field and interaction terms respectively.

\begin{theorem}[Necessity, modulo telescoping terms]
\label{thm:global-local}
Let $w$ be any one-site operator and let $H=\sum_jh_{j,j+1}$ be a nearest-neighbour Hamiltonian with arbitrary bond-dependent densities. The following equivalences assert the existence of a periodic collection of one-site operators $t_j$. For a uniform global additive operator $Q_+$ on a periodic chain of length $L\geq3$,
\begin{equation*}
 [H,Q_+]=0
 \quad\Longleftrightarrow\quad
 [\widehat h_{j,j+1},w_j+w_{j+1}]=0
 \quad\text{for every }j,\quad
 \widehat h_{j,j+1}=h_{j,j+1}+t_j-t_{j+1}.
\end{equation*}
For a staggered global additive operator $Q_-$ on an even periodic chain of length $L\geq4$,
\begin{equation*}
 [H,Q_-]=0
 \quad\Longleftrightarrow\quad
 [\widehat h_{j,j+1},w_j-w_{j+1}]=0
 \quad\text{for every }j,\quad
 \widehat h_{j,j+1}=h_{j,j+1}+t_j-t_{j+1}.
\end{equation*}
Here the $t_j$ are one-site operators that may vary with $j$. In either case, one may choose the canonical representative
\begin{equation}
 \widehat h_{j,j+1}
 =\bar cI+\kappa_{j,j+1}+\frac12(v_j+v_{j+1}),
 \label{eq:canonical-representative}
\end{equation}
with components defined in \eqref{eq:density-decomposition}-\eqref{eq:canonical-global-density}. A corresponding correction is
\begin{equation*}
 t_j=\tfrac12(\beta_j-\alpha_j)+\rho_jI_j,\qquad
 \rho_1=0,\qquad \rho_{j+1}-\rho_j=c_j-\bar c.
\end{equation*}
The scalar sequence is periodic because $\sum_j(c_j-\bar c)=0$. No Hermiticity or diagonalizability assumption on $w$ is required. If every density is Hermitian, these $t_j$ and $\widehat h_{j,j+1}$ are Hermitian as well.
\end{theorem}

\begin{proof}
First consider the uniform additive charge. Using \eqref{eq:canonical-global-density}, its commutator with $H$ is
\begin{equation*}
 0=[H,Q_+]=\sum_{j=1}^{L}[v_j,w_j]
 +\sum_{j=1}^{L}[\kappa_{j,j+1},w_j+w_{j+1}].
\end{equation*}
The first sum contains traceless one-site operators. Every term in the second sum has vanishing partial traces on both sites. Indeed, taking a partial trace over the site on which $w$ acts gives zero by cyclicity; taking it over the other site gives a commutator with the corresponding partial trace of $\kappa$, which vanishes. The linear independence established above therefore implies, for every $j$,
\begin{equation}
 [v_j,w_j]=0,\qquad
 [\kappa_{j,j+1},w_j+w_{j+1}]=0.
 \label{eq:uniform-canonical-conditions}
\end{equation}
Consequently, $\widehat h$ in \eqref{eq:canonical-representative} commutes with the bond sum. For every bond, direct subtraction shows that $\widehat h_{j,j+1}-h_{j,j+1}=t_j-t_{j+1}$ with the stated choice of $t_j$, so its periodic sum is exactly $H$.

For the staggered additive charge on an even ring, the same decomposition gives
\begin{equation*}
 0=[H,Q_-]=\sum_{j=1}^{L}(-1)^j[v_j,w_j]
 +\sum_{j=1}^{L}(-1)^j[\kappa_{j,j+1},w_j-w_{j+1}].
\end{equation*}
The signs are compatible with the closing bond because $L$ is even. The one-site terms are again traceless, and the bond terms again have vanishing partial traces. Independence applies equally to this signed sum, since every coefficient $(-1)^j$ is nonzero. Thus
\begin{equation}
 [v_j,w_j]=0,\qquad
 [\kappa_{j,j+1},w_j-w_{j+1}]=0.
 \label{eq:staggered-canonical-conditions}
\end{equation}
The representative \eqref{eq:canonical-representative} now commutes with the bond difference operator: its interaction part does so by \eqref{eq:staggered-canonical-conditions}, while each of its one-site terms commutes separately with $w$. This establishes necessity in both cases. Conversely, a telescoping change leaves $H$ unchanged, and the local relations imply global conservation by \eqref{eq:uniform-sufficiency} and \eqref{eq:staggered-sufficiency}. Finally, the partial-trace formulas preserve Hermiticity, proving the last assertion.
\end{proof}
This theorem is valid for even periodic chains in the staggered case. Odd periodic staggered chains retain the boundary obstruction discussed in Section~\ref{sec:local-scalene} and are outside the staggered equivalence proved here.

Theorem~\ref{thm:global-local} is closely related to
Lemma~II.2 of Moudgalya and Motrunich~\cite{MoudgalyaMotrunich2023},
whose proof in Appendix~B.1 uses operator expansions with
distinct supports to show that a local Hamiltonian invariant
under an on-site unitary symmetry admits a decomposition into
individually symmetric local terms. For Hermitian $w$, both
additive charges considered here generate on-site unitary
symmetries, placing our result within that framework.
Our formulation treats bond-dependent nearest-neighbour
densities explicitly, identifies their telescoping freedom,
and provides a canonical locally commuting representative.
The direct commutator proof also applies without assuming
that $w$ is Hermitian. This formulation supplies the
global-to-local correspondence needed for the exhaustive
scalene realizations established in Section~\ref{sec:local-scalene}.

The proof of Theorem~\ref{thm:global-local} identifies how global conservation constrains the Hamiltonian density. The interaction terms given by the two-site operators $\kappa$ must already lie in the relevant bond commutant, whereas the one-site contributions from adjacent bonds need only have a sum $v_j=\alpha_j+\beta_j$ that commutes with $w_j$. Redistributing each $v_j$ equally between its two adjacent bonds produces locally commuting densities. In particular, the theorem does not assert that the original Hamiltonian density must commute with the bond generator. For example, a density $h_{j,j+1}=\alpha_j-\alpha_{j+1}$ has zero periodic sum for any $\alpha$, although its local commutator can be nonzero.

The canonical map $h\mapsto\widehat h$ is linear. Thus families of globally symmetric densities can be put into local bond-commuting form with the same coupling dependence, after which the construction in Section~\ref{sec:local-scalene} supplies a fixed pair $A,C$ for the entire family.

The quantities $\bar c$, $\{v_j\}$ and $\{\kappa_{j,j+1}\}$ are invariant under telescoping changes and uniquely determine $H$. By support independence, two collections of nearest-neighbour densities give the same periodic Hamiltonian at fixed $L\geq3$ if and only if these quantities agree. Equivalently, their difference is $t_j-t_{j+1}$ for a periodic collection of one-site operators. Thus \eqref{eq:canonical-representative} is the unique representative with a common scalar part and with each physical one-site term shared equally between its adjacent bonds. The symmetry classes are characterized by arbitrary $\bar c$, traceless $v_j$ satisfying $[v_j,w_j]=0$, and zero-partial-trace interactions $\kappa_{j,j+1}$ satisfying the appropriate local relation. These are independent linear conditions at each site and bond in any finite local dimension. In particular, the interaction operators $\kappa_{j,j+1}$ may differ in their operator structure from bond to bond. Such differences survive the telescoping quotient, since a telescoping term changes only one-site and scalar contributions.

Our main result can now be summarized as follows.
By combining Theorem~\ref{thm:global-local} with Theorems~\ref{thm:sum-bond} and \ref{thm:difference-bond}, we see that every Hamiltonian in the specified additive-symmetry classes admits a scalene realization. The canonical density $\widehat h$ in \eqref{eq:canonical-representative} differs from the original density by a telescoping term and determines the regular member of the triple through
\begin{equation*}
 \check B_{j,j+1}(u)=I+u\widehat h_{j,j+1},\qquad
 B_{j,j+1}(u)=P_{j,j+1}\check B_{j,j+1}(u).
\end{equation*}
The corresponding $A$ and $C$ operators are supplied by Theorem~\ref{thm:sum-bond} or Theorem~\ref{thm:difference-bond}, according to the symmetry. The resulting triple satisfies the scalene equation \eqref{eq:intro-scalene} for every $u$. The regular member generates the original Hamiltonian, since the density change telescopes, and the non-regular monodromy contains its additive charge as established in Section~\ref{sec:local-scalene}. 


\section{Examples}
\label{sec:examples}
Throughout this section the scalar coefficients in each displayed density may be chosen independently on every bond. However, we suppress their bond labels for readability. The examples are constructed by the following steps:
\begin{enumerate}
    \item Choose a one-site operator $w$. For the examples considered here we choose this to be the third component of a spin $J$ system, namely $w=S^z$.
    \item We then construct the commutant of the sum- and difference-bond operators constructed using the chosen $w$. In our case this is done using the theory of quantum angular momentum.
    \item The commutant completely specifies the Hamiltonian density up to the telescoping term and thus also fixes the total Hamiltonian on the closed chain.
    \item The regular member of the triple, $B$ is then obtained from this Hamiltonian density. The non-regular members, $A$ and $C$, are specified by Theorems \ref{thm:sum-bond} and \ref{thm:difference-bond} for the two global additive symmetries, respectively.
\end{enumerate}


The sum-bond commutant can be determined using the addition of angular momenta. We take each local Hilbert space $\mathcal H_J$ to carry the spin-$J$ irreducible representation, with $J=\tfrac12,1,\tfrac32,\ldots$ and dimension $d=2J+1$. Setting $\hbar=1$, define the total angular momentum on a bond by
\begin{equation*}
 \mathbf S_{\mathrm{tot}}=\mathbf S_j+\mathbf S_{j+1},
 \qquad S_{\mathrm{tot}}^z=S_j^z+S_{j+1}^z.
\end{equation*}
Adding the two spins gives the decomposition
\begin{equation*}
 \mathcal H_J\otimes\mathcal H_J
 =\bigoplus_{\ell=0}^{2J}\mathcal H_\ell,
\end{equation*}
in which each total-spin representation $\ell$ occurs once. The coupled basis states $|\ell,M\rangle$ satisfy
\begin{equation*}
 \mathbf S_{\mathrm{tot}}^2|\ell,M\rangle
 =\ell(\ell+1)|\ell,M\rangle,
 \qquad S_{\mathrm{tot}}^z|\ell,M\rangle=M|\ell,M\rangle,
\end{equation*}
where $M=-\ell,-\ell+1,\ldots,\ell$. The bond sum is the third component of total angular momentum, so its distinct eigenvalues are
\begin{equation*}
 M=-2J,-2J+1,\ldots,2J.
\end{equation*}
These values are integers even when $J$ is half-integer, since the two spins being added are equal.

For a fixed $M$, the state $|\ell,M\rangle$ occurs for each $\ell=|M|,|M|+1,\ldots,2J$. The corresponding eigenspace and its dimension are therefore
\begin{equation*}
 \mathcal K_M=\operatorname{span}\{|\ell,M\rangle:\ell=|M|,\ldots,2J\},
 \qquad n_M=\dim\mathcal K_M=2J+1-|M|.
\end{equation*}
Thus, in order of increasing $M$, the eigenvalue multiplicities are $1,2,\ldots,2J,2J+1,2J,\ldots,2,1$. Their sum accounts for the full two-spin Hilbert space:
\begin{equation}
 \sum_{M=-2J}^{2J}n_M
 =(2J+1)+2\sum_{r=1}^{2J}r
 =(2J+1)^2.
 \label{eq:angular-hilbert-dimension}
\end{equation}

To determine the commutant, consider the matrix elements of the braided operator $\check B_{j,j+1}$ in this coupled basis. They obey
\begin{equation*}
 \langle\ell,M|[S_{\mathrm{tot}}^z,\check B_{j,j+1}]|\ell',M'\rangle
 =(M-M')\langle\ell,M|\check B_{j,j+1}|\ell',M'\rangle.
\end{equation*}
Consequently, $\check B_{j,j+1}$ commutes with the bond sum precisely when it preserves $M$. Within each eigenspace $\mathcal K_M$, all matrix elements are allowed. The general member of the commutant therefore has the form
\begin{equation*}
 \check B_{j,j+1}
 =\sum_{M=-2J}^{2J}\sum_{\ell,\ell'=|M|}^{2J}
 b^{(M)}_{\ell\ell'}|\ell,M\rangle\langle\ell',M|,
\end{equation*}
where $b^{(M)}_{\ell\ell'}$ are arbitrary complex coefficients. Equivalently,
\begin{equation}
 \mathfrak C_+(S^z)
 =\bigoplus_{M=-2J}^{2J}\operatorname{End}(\mathcal K_M)
 \simeq\bigoplus_{M=-2J}^{2J}
 \operatorname{Mat}_{\,2J+1-|M|}(\mathbb C).
 \label{eq:angular-commutant-blocks}
\end{equation}
Each eigenspace of dimension $n_M$ contributes $n_M^2$ independent operators. Hence the dimension of the commutant is the sum of the squares of the eigenvalue multiplicities:
\begin{align}
 \dim_{\mathbb C}\mathfrak C_+(S^z)
 &=\sum_{M=-2J}^{2J}(2J+1-|M|)^2\notag\\
 &=(2J+1)^2+2\sum_{r=1}^{2J}r^2\notag\\
 &=\frac{(2J+1)(8J^2+8J+3)}{3}
 =\frac{d(2d^2+1)}{3},\qquad d=2J+1.
 \label{eq:angular-commutant-dimension}
\end{align}
It is important to check that this dimension is an integer to ensure the validity of this result. This is indeed true as seen from the following argument. The dimension $d$ can take one of three values $3k$, $3k+1$ or $3k+2$, for any positive integer $k$. In each of these cases $\dim_{\mathbb C}\mathfrak C_+(S^z)$ becomes $k(18k^2+1)$, $(3k+1)(6k^2+4k+1)$ and $(3k+2)(6k^2+8k+3)$, respectively. As each of them is a positive integer we have our desired result.

The distinction between \eqref{eq:angular-hilbert-dimension} and \eqref{eq:angular-commutant-dimension} is that the former counts states, while the latter counts independent operators acting within the fixed-magnetization eigenspaces.

Conservation of total magnetization allows transitions between different total-spin multiplets: $|\ell,M\rangle\langle\ell',M|$ belongs to the commutant even when $\ell\ne\ell'$. Commutation with $\mathbf S_{\mathrm{tot}}^2$, or invariance under all total-spin components, would impose additional restrictions on such transitions. The algebra \eqref{eq:angular-commutant-blocks} retains every two-site operator that conserves $S_j^z+S_{j+1}^z$, as required for the sum-bond scalene construction.

The difference bond admits a parallel analysis by rotating the second spin. Let $g=\exp(-i\pi S^x)$ be a rotation by $\pi$ about the $x$-axis in the spin-$J$ representation. Since $gS^zg^\dagger=-S^z$, its action on the second site gives
\begin{equation}
 S_j^z-S_{j+1}^z
 =g_{j+1}(S_j^z+S_{j+1}^z)g_{j+1}^\dagger.
 \label{eq:angular-difference-rotation}
\end{equation}
Thus the two bond generators are unitarily equivalent for every spin $J$. To relate this to addition of angular momenta, the rotated second spin has components
\[
 \widetilde{\mathbf S}_{j+1}
 =g_{j+1}\mathbf S_{j+1}g_{j+1}^\dagger
 =(S_{j+1}^x,-S_{j+1}^y,-S_{j+1}^z).
\]
These satisfy the usual angular-momentum commutation relations and still carry spin $J$. The difference bond is the third component of $\mathbf S_j+\widetilde{\mathbf S}_{j+1}$. Accordingly, define the rotated coupled states
\begin{equation}
 |\ell,M\rangle_-=g_{j+1}|\ell,M\rangle,
 \qquad
 (S_j^z-S_{j+1}^z)|\ell,M\rangle_-=M|\ell,M\rangle_-.
 \label{eq:angular-difference-states}
\end{equation}
The label $\ell$ now refers to the total spin formed with the rotated second spin, rather than to the original $\mathbf S_{\mathrm{tot}}^2$.

The eigenvalues remain $M=-2J,-2J+1,\ldots,2J$. For each $M$, the eigenspace is
\begin{equation*}
 \mathcal K_M^-=g_{j+1}\mathcal K_M
 =\operatorname{span}\{|\ell,M\rangle_-:\ell=|M|,\ldots,2J\},
 \qquad \dim\mathcal K_M^-=2J+1-|M|.
\end{equation*}
One can also count this multiplicity directly in the product basis: the difference bond assigns eigenvalue $m_1-m_2$ to $|m_1,m_2\rangle$, and exactly $2J+1-|M|$ pairs obey $m_1-m_2=M$. As for the sum bond, the multiplicities add to $(2J+1)^2$.

The commutant consists of arbitrary operators within these fixed-difference eigenspaces. In the rotated coupled basis, its general member is
\begin{equation*}
 \check B_{j,j+1}
 =\sum_{M=-2J}^{2J}\sum_{\ell,\ell'=|M|}^{2J}
 b^{(M,-)}_{\ell\ell'}|\ell,M\rangle_-\,{}_-\langle\ell',M|,
\end{equation*}
with unrestricted complex coefficients. Hence
\begin{align*}
 \mathfrak C_-(S^z)
 &=g_{j+1}\mathfrak C_+(S^z)g_{j+1}^\dagger
 =\bigoplus_{M=-2J}^{2J}\operatorname{End}(\mathcal K_M^-),\\
 \dim_{\mathbb C}\mathfrak C_-(S^z)
 &=\sum_{M=-2J}^{2J}(2J+1-|M|)^2
 =\frac{d(2d^2+1)}{3},\qquad d=2J+1.
\end{align*}
The two commutants therefore have the same block sizes and dimension, but preserve different combinations of the physical magnetizations. These are local two-site statements and do not depend on chain length. When difference-commuting densities are assembled into an untwisted periodic Hamiltonian, the even-length condition for the staggered global charge remains as established in Section~\ref{sec:local-scalene}.

\subsection*{\texorpdfstring{Spin-$\tfrac12$}{Spin-1/2}}

We now give the most general densities in the coupled angular-momentum basis, allowing arbitrary complex coefficients. This use of an $SU(2)$ coupled basis does not impose $SU(2)$ invariance on the Hamiltonian. For two spins $\tfrac12$, the coupled states are
\begin{equation*}
 |t_+\rangle=|\uparrow\uparrow\rangle,\quad
 |t_0\rangle=\frac{|\uparrow\downarrow\rangle+|\downarrow\uparrow\rangle}{\sqrt2},\quad
 |s\rangle=\frac{|\uparrow\downarrow\rangle-|\downarrow\uparrow\rangle}{\sqrt2},\quad
 |t_-\rangle=|\downarrow\downarrow\rangle.
\end{equation*}
The three $t$ states have total spin one and $|s\rangle$ has total spin zero. From the general spin $J$ analysis we see that the most general sum-commuting density is
\begin{equation}
 h^{(+)}=\epsilon_+|t_+\rangle\langle t_+|
 +\epsilon_-|t_-\rangle\langle t_-|
 +u|t_0\rangle\langle t_0|+v|s\rangle\langle s|
 +\xi|t_0\rangle\langle s|+\zeta|s\rangle\langle t_0|.
 \label{eq:half-general-coupled-density}
\end{equation}
All six coefficients are independent complex numbers. The density is Hermitian precisely when $\epsilon_\pm,u,v$ are real and $\zeta=\xi^*$. In particular, the terms proportional to $\xi$ and $\zeta$ mix the singlet and triplet while preserving their common magnetization $M=0$.

Suppressing the site labels on the two-site matrices, define
\[
 p=\frac{u+v+\xi+\zeta}{2},\quad
 q=\frac{u+v-\xi-\zeta}{2},\quad
 f=\frac{u-v-\xi+\zeta}{2},\quad
 z=\frac{u-v+\xi-\zeta}{2}.
\]
In the product basis $(|\uparrow\uparrow\rangle,|\uparrow\downarrow\rangle,
 |\downarrow\uparrow\rangle,|\downarrow\downarrow\rangle)$, the two most general densities are
\begin{equation*}
 h^{(+)}=
 \begin{pmatrix}
 \epsilon_+&0&0&0\\
 0&p&f&0\\
 0&z&q&0\\
 0&0&0&\epsilon_-
 \end{pmatrix},\qquad
 h^{(-)}=
 \begin{pmatrix}
 p&0&0&f\\
 0&\epsilon_+&0&0\\
 0&0&\epsilon_-&0\\
 z&0&0&q
 \end{pmatrix}.
\end{equation*}
The second follows by flipping the second spin, $h^{(-)}=\sigma_{j+1}^x h^{(+)}\sigma_{j+1}^x$, and exhausts the difference commutant. Equivalently, its coupled-basis expression is \eqref{eq:half-general-coupled-density} with every ket replaced by its rotated counterpart from \eqref{eq:angular-difference-states}. A common phase from the rotation cancels in each outer product. The parameters can be chosen independently for the sum and difference families; the displayed correspondence identifies one density in each. In matrix notation, Hermiticity requires $p,q,\epsilon_\pm\in\mathbb R$ and $z=f^*$. The six complex entries give the dimension of each commutant directly.

The corresponding Hamiltonians are $H_\pm=\sum_jh^{(\pm)}_{j,j+1}$. To realize them as scalene triples, set $S^z=\sigma^z/2$, $s=S^+=(\sigma^x+i\sigma^y)/2$ and $r=\sigma^z$. For any fixed $\lambda\ne0$, the triples are
\begin{align*}
 B^{(\pm)}_{13}(u)&=P_{13}(I+u h^{(\pm)}_{13}),\\
 A^{(+)}_{12}&=I+\lambda S_1^zS_2^+,& C^{(+)}_{23}&=I-\lambda S_2^+S_3^z,\\
 A^{(-)}_{12}&=\sigma_2^z+\lambda S_1^zS_2^+,& C^{(-)}_{23}&=\sigma_2^z+\lambda S_2^+S_3^z.
\end{align*}
These triples satisfy the scalene equation for arbitrary complex density parameters and any $\lambda\ne0$. The alternative $s=S^-$ with the same $r=\sigma^z$ also works with the same $B$, illustrating the non-uniqueness of the triples discussed in Section~\ref{sec:local-scalene}. Their regular members generate $H_\pm$, while their non-regular monodromies contain $Q_\pm$, with even length required for the staggered periodic chain.

For an explicit spatially nonuniform spin-$\tfrac12$ example, partition the bonds into three nonempty sets $\mathcal I_1,\mathcal I_2,\mathcal I_3$, with no imposed spatial pattern, and take
\begin{align*}
 H_+^{(1/2)}={}&\sum_{j\in\mathcal I_1}J_j
 (S_j^+S_{j+1}^-+S_j^-S_{j+1}^+)
 +\sum_{j\in\mathcal I_2}\Delta_j S_j^zS_{j+1}^z\\
 &+i\sum_{j\in\mathcal I_3}D_j
 (S_j^+S_{j+1}^- - S_j^-S_{j+1}^+).
\end{align*}
The three bond types are respectively exchange, longitudinal interaction and antisymmetric exchange. They are linearly independent and have vanishing partial traces, so their distinction cannot be removed by scalar rescaling or telescoping terms. Every bond commutes with $S_j^z+S_{j+1}^z$, and hence $[H_+^{(1/2)},Q_+]=0$. Real $J_j,\Delta_j,D_j$ give a Hermitian example; complex values are also allowed. On an even ring, a corresponding staggered example is
\begin{align*}
 H_-^{(1/2)}={}&\sum_{j\in\mathcal I_1}J_j
 (S_j^+S_{j+1}^++S_j^-S_{j+1}^-)
 -\sum_{j\in\mathcal I_2}\Delta_j S_j^zS_{j+1}^z\\
 &+i\sum_{j\in\mathcal I_3}D_j
 (S_j^+S_{j+1}^+-S_j^-S_{j+1}^-).
\end{align*}
Each density is obtained by conjugating its sum-bond counterpart with $\sigma_{j+1}^x$, so it commutes with $S_j^z-S_{j+1}^z$. Both examples use the same respective $A,C$ displayed above, with $\check B_{j,j+1}(u)=I+u h_{j,j+1}^{(\pm)}$ independently on each bond and the monodromy assignment in \eqref{eq:inhomogeneous-derivative}.

\subsection*{Spin-1}

For two spin-1 sites, the coupled representations have total spin $\ell=2,1,0$. We denote total-angular-momentum eigenstates by $|\ell,M\rangle_{\mathrm c}$ and uncoupled product states by $|m_1,m_2\rangle_{\mathrm p}$. In the latter notation, $m_1$ and $m_2$ are the eigenvalues of $S_j^z$ and $S_{j+1}^z$, respectively, with both one-site spins fixed at $J=1$. Choosing standard relative phases, the change of basis is
\begin{align}
 |2,\pm2\rangle_{\mathrm c}&=|\pm1,\pm1\rangle_{\mathrm p},\notag\\
 |2,\pm1\rangle_{\mathrm c}&=\frac{|\pm1,0\rangle_{\mathrm p}+|0,\pm1\rangle_{\mathrm p}}{\sqrt2},\notag\\
 |1,1\rangle_{\mathrm c}&=\frac{|1,0\rangle_{\mathrm p}-|0,1\rangle_{\mathrm p}}{\sqrt2},\qquad
 |1,-1\rangle_{\mathrm c}=\frac{|0,-1\rangle_{\mathrm p}-|-1,0\rangle_{\mathrm p}}{\sqrt2},\notag\\
 |2,0\rangle_{\mathrm c}&=\frac{|1,-1\rangle_{\mathrm p}+2|0,0\rangle_{\mathrm p}+|-1,1\rangle_{\mathrm p}}{\sqrt6},\notag\\
 |1,0\rangle_{\mathrm c}&=\frac{|1,-1\rangle_{\mathrm p}-|-1,1\rangle_{\mathrm p}}{\sqrt2},\qquad
 |0,0\rangle_{\mathrm c}=\frac{|1,-1\rangle_{\mathrm p}-|0,0\rangle_{\mathrm p}+|-1,1\rangle_{\mathrm p}}{\sqrt3}.
 \label{eq:one-coupled-states-explicit}
\end{align}
Thus $|0,0\rangle_{\mathrm c}$ is the total-spin singlet, whereas $|0,0\rangle_{\mathrm p}$ has zero magnetization on each site. The subscripts $\mathrm c$ and $\mathrm p$ distinguish the two bases throughout this example.

Order the coupled basis by magnetization as
\begin{equation}
 \bigl(|2,2\rangle_{\mathrm c}; |2,1\rangle_{\mathrm c},|1,1\rangle_{\mathrm c};
 |2,0\rangle_{\mathrm c},|1,0\rangle_{\mathrm c},|0,0\rangle_{\mathrm c};
 |2,-1\rangle_{\mathrm c},|1,-1\rangle_{\mathrm c}; |2,-2\rangle_{\mathrm c}\bigr).
 \label{eq:one-ordered-coupled-basis}
\end{equation}
The most general sum-commuting density has the explicit $9\times9$ block form
\begin{equation}
 h^{(+)}=\operatorname{diag}(\epsilon_2,K_1,K_0,K_{-1},\epsilon_{-2}),
 \quad
 K_{\pm1}=\begin{pmatrix}
 \gamma^{(\pm1)}_{22}&\gamma^{(\pm1)}_{21}\\
 \gamma^{(\pm1)}_{12}&\gamma^{(\pm1)}_{11}
 \end{pmatrix},
 \quad
 K_0=\begin{pmatrix}
 \gamma^{(0)}_{22}&\gamma^{(0)}_{21}&\gamma^{(0)}_{20}\\
 \gamma^{(0)}_{12}&\gamma^{(0)}_{11}&\gamma^{(0)}_{10}\\
 \gamma^{(0)}_{02}&\gamma^{(0)}_{01}&\gamma^{(0)}_{00}
 \end{pmatrix}.
 \label{eq:one-general-block-density}
\end{equation}
Equivalently, as an operator in the angular-momentum basis,
\begin{equation*}
 h^{(+)}=\epsilon_2|2,2\rangle_{\mathrm c}\,{}_{\mathrm c}\langle2,2|
 +\epsilon_{-2}|2,-2\rangle_{\mathrm c}\,{}_{\mathrm c}\langle2,-2|
 +\sum_{M=-1}^{1}\sum_{\ell,\ell'=|M|}^{2}
 \gamma^{(M)}_{\ell\ell'}|\ell,M\rangle_{\mathrm c}\,{}_{\mathrm c}\langle\ell',M|.
\end{equation*}
At $M=0$, the state with $\ell=0$ is the coupled singlet in \eqref{eq:one-coupled-states-explicit}. The two $\epsilon$ coefficients and all entries of the three $K$ blocks are arbitrary complex numbers, giving $2+4+9+4=19$ parameters. Hermiticity holds precisely when $\epsilon_{\pm2}$ are real and $K_M=K_M^\dagger$, or equivalently $\gamma^{(M)}_{\ell\ell'}=(\gamma^{(M)}_{\ell'\ell})^*$. This locus has nineteen real parameters.

For the difference bond, define the spin-flip operator
\begin{equation*}
 F=|1\rangle\langle-1|+|0\rangle\langle0|+|-1\rangle\langle1|,
 \qquad F S^z F=-S^z,\qquad F^2=I.
\end{equation*}
For spin one, $F=-g$ with $g$ from \eqref{eq:angular-difference-rotation}. Thus the most general difference-commuting density is
\begin{equation*}
 h^{(-)}=F_{j+1}h^{(+)}F_{j+1}.
\end{equation*}
Its $9\times9$ matrix is exactly \eqref{eq:one-general-block-density} in the basis obtained by applying $F_{j+1}$ to each state in \eqref{eq:one-ordered-coupled-basis}. Explicitly, this replaces every product ket $|m_1,m_2\rangle_{\mathrm p}$ in \eqref{eq:one-coupled-states-explicit} by $|m_1,-m_2\rangle_{\mathrm p}$. This specifies all matrix entries without expanding a dense change-of-basis expression. The blocks now preserve magnetization difference, and their Hermiticity conditions are unchanged. As before, the two families allow independent choices of all coefficients.

Finally, writing $E_{mn}=|m\rangle\langle n|$ for one-site matrix units, choose $s=E_{10}$ and $r=I-2E_{00}$. Since $s^2=0$, $r^2=I$ and $rs=-sr$, for any fixed $\lambda\ne0$ the triples are
\begin{align*}
 B^{(\pm)}_{13}(u)&=P_{13}(I+u h^{(\pm)}_{13}),\\
 A^{(+)}_{12}&=I+\lambda S_1^z(E_{10})_2,& C^{(+)}_{23}&=I-\lambda (E_{10})_2S_3^z,\\
 A^{(-)}_{12}&=I-2(E_{00})_2+\lambda S_1^z(E_{10})_2,&
 C^{(-)}_{23}&=I-2(E_{00})_2+\lambda (E_{10})_2S_3^z.
\end{align*}
An alternative auxiliary realization uses $s=E_{1,-1}$ and $r=I-2E_{-1,-1}$, with the same regular member $B$. The periodic Hamiltonians $H_\pm=\sum_jh^{(\pm)}_{j,j+1}$ and their charges follow from these triples as in Section~\ref{sec:local-scalene}. Together with Theorem~\ref{thm:global-local}, the unrestricted complex families above exhaust the Hamiltonians with the respective additive symmetry modulo local telescoping changes, subject to the stated chain-length conditions.

The spin-1 space allows further changes of interaction type. Using an arbitrary bond partition as above, an explicit example is
\begin{align}
 H_+^{(1)}={}&\sum_{j\in\mathcal I_1}J_j
 (S_j^+S_{j+1}^-+S_j^-S_{j+1}^+)\nonumber\\
 &+\sum_{j\in\mathcal I_2}K_j
 \bigl((S_j^+)^2(S_{j+1}^-)^2+(S_j^-)^2(S_{j+1}^+)^2\bigr)
 +\sum_{j\in\mathcal I_3}D_j(S_j^z)^2(S_{j+1}^z)^2. \nonumber\\
 =&\sum_{j\in\mathcal I_1} J_j\mathcal{M}_j^{(1)}
 +\sum_{j\in\mathcal I_2} K_j\mathcal{M}_j^{(2)}
  +\sum_{j\in\mathcal I_3}D_j\mathcal{V}_j.
  \label{eq:spin1-example-Hamiltonian}
\end{align}
We introduce the bond operators
\begin{align*}
\mathcal{M}_j^{(1)}
&=S_j^+S_{j+1}^-+S_j^-S_{j+1}^+, \quad
\mathcal{M}_j^{(2)}
=(S_j^+)^2(S_{j+1}^-)^2
 +(S_j^-)^2(S_{j+1}^+)^2, \quad
\mathcal{V}_j
=(S_j^z)^2(S_{j+1}^z)^2.
\label{eq:spin1-bond-operators}
\end{align*}
The three bond operators have distinct actions in the two-site product basis. The single-magnetization exchange operator $\mathcal{M}_j^{(1)}$ connects the pairs
$|1,0\rangle_{\mathrm p}\leftrightarrow
|0,1\rangle_{\mathrm p}$ and
$|0,-1\rangle_{\mathrm p}\leftrightarrow
|-1,0\rangle_{\mathrm p}$, as well as
$|1,-1\rangle_{\mathrm p}\leftrightarrow
|0,0\rangle_{\mathrm p}$ and
$|0,0\rangle_{\mathrm p}\leftrightarrow
|-1,1\rangle_{\mathrm p}$.
Each of these transitions has matrix element $2$. In contrast, the double-magnetization exchange operator $\mathcal{M}_j^{(2)}$ acts only between
$|1,-1\rangle_{\mathrm p}$ and $|-1,1\rangle_{\mathrm p}$, with matrix element $4$. The diagonal operator $\mathcal V_j$ leaves all product states unchanged and has eigenvalue $1$ when both local magnetizations are nonzero,
and $0$ otherwise. Each bond operator commutes with the local magnetization sum $(S_j^z+S_{j+1}^z)$, while their two-site interaction components are linearly independent. In the last term the zero-partial-trace component is proportional to $((S_j^z)^2-\tfrac23I_j)((S_{j+1}^z)^2-\tfrac23I_{j+1})$. Thus these distinct bond types remain distinct modulo telescoping terms. 

To make the connection with the general density in \eqref{eq:one-general-block-density} explicit, consider a representative two-site Hamiltonian density
$h_j=J\mathcal{M}_j^{(1)}+K\mathcal{M}_j^{(2)}+D \mathcal{V}_j$. Using the spin-1 ladder-operator matrix elements and the coupled states in \eqref{eq:one-coupled-states-explicit}, we obtain the following $M=0$ block in the ordered basis $(|2,0\rangle_{\mathrm c},
|1,0\rangle_{\mathrm c},
|0,0\rangle_{\mathrm c})$:
\begin{equation*}
K_0=\frac{1}{3}
\begin{pmatrix}
8J+4K+D & 0 & \sqrt2(2J+4K+D)\\
0 & -12K+3D & 0\\
\sqrt2(2J+4K+D) & 0 & -8J+8K+2D
\end{pmatrix}.
\label{eq:spin1-example-K0}
\end{equation*}
The remaining magnetization sectors give
\begin{equation*}
K_{\pm 1}=
\begin{pmatrix}
2J&0\\
0&-2J
\end{pmatrix},
\qquad
\epsilon_{\pm 2}=D.
\label{eq:spin1-example-Kpm}
\end{equation*}
The above matrices provide an explicit three-parameter subfamily of the general $19$-parameter magnetization-conserving density. In particular, the nonzero off-diagonal elements of $K_0$ couple the total-spin
states $|2,0\rangle_{\mathrm c}$ and $|0,0\rangle_{\mathrm c}$, demonstrating that conservation of total magnetization does not necessarily imply conservation of total angular momentum.

Furthermore, the vanishing matrix elements between the $\ell=1$ and $\ell=0,2$ states follow from an additional site-exchange symmetry of the chosen bond operators. Indeed, under the permutation $P_{j,j+1}$ of the two sites,
the coupled spin-1 states satisfy 
\begin{equation*}
P_{j,j+1}|\ell,M\rangle_{\mathrm c}
=(-1)^{2-\ell}|\ell,M\rangle_{\mathrm c}.
\end{equation*}
Thus the $\ell=0,2$ states are symmetric, whereas the $\ell=1$ states are antisymmetric. Since the three bond operators $\mathcal{M}_j^{(1)}$, $\mathcal{M}_j^{(2)}$, and $\mathcal{V}_j$ commute with $P_{j,j+1}$, their matrix elements between states of opposite exchange parity vanish. This explains the decoupling of the $\ell=1$ state in the $M=0$ block, while mixing between the $\ell=0$ and $\ell=2$ states remains allowed. Such an additional restriction is not imposed on the general magnetization-conserving density in \eqref{eq:one-general-block-density}.

For the staggered example, set
\begin{equation*}
 H_-^{(1)}=\sum_j F_{j+1}h_{j,j+1}^{(+)}F_{j+1}.
\end{equation*}
This replaces opposite-sign ladder pairs in the first two sums by equal-sign pairs and leaves the last sum unchanged. It conserves $Q_-$ for even $L$. Real coefficients give Hermitian Hamiltonians in both cases. The spin-1 triples above realize these independently chosen densities through their regular $B$ members, without changing $A,C$ or introducing site-dependent spectral shifts.

\section{Discussion}
\label{sec:discussion-outlook}

The construction presented here works in two complementary directions. Given $w$, its sum- or difference-bond commutant supplies Hamiltonian densities and their scalene realizations, with the freedom in $A$ and $C$ described in Section~\ref{sec:local-scalene}. Conversely, a Hamiltonian with the corresponding global additive symmetry admits a bond-commuting density $\widehat h_{j,j+1}$ after removing a telescoping term. This adjustment redistributes the one-site, magnetic-field-like terms without changing the periodic Hamiltonian as discussed in Section~\ref{sec:necessity}. Both directions allow arbitrary bond dependence: the regular monodromy uses site-dependent $B$ operators with the same $A,C$. This permits different interactions and different charge-block matrices on different bonds, as the explicit examples demonstrate, while retaining one common algebraic realization of the charge. Together, these results exhaust the Hamiltonians under consideration modulo telescoping changes, with even length required in the staggered case. For a fixed $w$, the same $A,C$ also accommodate arbitrary symmetry-preserving nearest-neighbour perturbations after this adjustment of density. Our spin-$\tfrac12$ and spin-1 examples give the full density families and their Hermitian subfamilies for the particular choice $w=S^z$.

Homogeneous choices within the families in Section~4 contain familiar integrable subfamilies: the spin-$\tfrac12$ sum-bond family includes the $XXZ$ chain and its isotropic $XXX$ limit~\cite{Bethe1931,KBI}, while the spin-1 family includes the Takhtajan-Babujian density $h_{j,j+1}=\mathbf S_j\cdot\mathbf S_{j+1}-(\mathbf S_j\cdot\mathbf S_{j+1})^2$, up to normalization and an additive constant~\cite{Takhtajan1982,Babujian1982}. On even periodic chains, alternating rotations by $\pi$ about the $x$-axis map these models to integrable members of the difference-bond families. Their integrability imposes additional structure beyond the additive symmetry required by our scalene construction. Now we discuss two further aspects of the construction.

\begin{enumerate}
\item The scalene solution space realizes a finite-dimensional double-commutant pair~\cite{vonNeumann1930,BratteliRobinson}. For $w=w^\dagger$, let $D_\pm=w_j\pm w_{j+1}$ and denote with $\mathfrak D_\pm=\operatorname{Alg}^\dagger\{I,D_\pm\}$, the unital, adjoint-closed algebra generated by the unit or identity element $I$ and $D_\pm$. The allowed $\check B$ operators form $\mathfrak C_\pm(w)=\mathfrak D_\pm'$, where $\mathfrak D'$ denotes the commutant or centralizer of the algebra $\mathfrak D$. Then we have,
\begin{equation*}
 \mathfrak C_\pm(w)'=\mathfrak D_\pm''=\mathfrak D_\pm.
\end{equation*}
Indeed, $\mathfrak D_\pm'$ contains arbitrary matrices on each eigenspace of $D_\pm$; its commutant is scalar on each eigenspace and hence consists of polynomials in $D_\pm$. Similarly on the chain, the unital, adjoint-closed algebra $\mathfrak A_L$ generated by allowed bond Hamiltonian densities obeys
\begin{equation*}
 (\mathfrak A_L')'=\mathfrak A_L,\qquad Q_\pm\in\mathfrak A_L',
\end{equation*}
with the stated length restrictions. The monodromy supplies $Q_\pm$ explicitly, but need not generate the full commutant. This connects our construction to the bond-algebra approach of Cobanera, Nussinov and Ortiz~\cite{NussinovOrtiz2009, CobaneraOrtizNussinov2010,CobaneraOrtizNussinov2011}, and the commutant framework for symmetries and Hilbert space fragmentation~\cite{MoudgalyaMotrunich2022,MoudgalyaMotrunich2023}.

The above statements find a natural place in the context of algebraic quantum theory~\cite{HaagKastler,BratteliRobinson}, where the bicommutant of a unital, adjoint-closed algebra equals its strong-operator closure. In finite dimension, the algebra is already closed and thus is not so interesting. However, extending our construction to infinite chains requires specifying this closure and can possibly show some interesting features. These become more involved for non-Hermitian $w$, as now adjoint closure also needs attention. In particular, adjoining $D_\pm^\dagger$ can change the commutant. All this will help us formulate the Haag duality for such spin chains in the infinite lattice. For this we  would additionally need to identify a region's commutant with its complementary-region algebra, a locality property beyond the bicommutant identity. Failure of this duality property could indicate signatures of some global properties. These considerations are left for a future work.

\item Apart from the sum- and difference-bond operators considered in this work, we can also consider product bonds of the form $w_jw_{j+1}$. It is easy to see that when the Hamiltonian density commutes with such product-bond operators then the full Hamiltonian $H=\sum\limits_j~h_{j,j+1}$ conserves the global operator $G=\prod_jw_j$. The latter represents an on-site symmetry when $w$ is unitary. Appendix~\ref{app:product-bonds} gives a scalene realization for invertible $w$ and derives $[H,G]=0$. Here the product charge can survive the auxiliary trace. An exhaustive global-to-local correspondence remains to be established.
\end{enumerate}

\section{Outlook}
\label{sec:outlook}
The construction in this work lays the path to classify all nearest-neighbor Hamiltonians with a global additive charge.
This can be done by classifying the one-site generators $w$ and constructing their bond commutants. That would exhaust the nearest-neighbour Hamiltonians, including spatially disordered ones, with nontrivial uniform or staggered additive symmetry, modulo telescoping terms and with the stated periodic length restrictions. Each $w$ determines a local commutant from which the density on every bond can be chosen independently. The resulting classification therefore covers arbitrary spatial arrangements of all allowed interaction types, beyond scalar disorder applied to a fixed spin-chain density. Thus classifying $w$ supplies all these families without a translation-invariance assumption. Furthermore these Hamiltonians and their conserved charges can be obtained from the regular and non-regular members of scalene triples. 

We will take a brief look at the main steps involved in such a classification. Consider the following set of relations
\begin{equation*}
 w\sim\alpha XwX^{-1}+\beta I,
 \qquad \alpha\ne0,\quad\alpha,\beta\in\mathbb C,\quad
 X\in\mathrm{GL}(\mathcal H).
\end{equation*}
The constant shifts and rescalings leave the bond commutants unchanged and so we consider $w$'s up to these equivalences. On the other hand similarity transformations conjugate the bond commutants by $X_jX_{j+1}$. Jordan forms therefore organize the possibilities for $w$. Writing $J_n(\mu)$ for a size-$n$ Jordan block with unit superdiagonal and $[\mu]$ for a one-dimensional block, the complete spin-$\tfrac12$ list is
\begin{equation*}
 0,\qquad\operatorname{diag}(0,1),\qquad J_2(0),
\end{equation*}
and the spin-1 list is
\begin{equation*}
\begin{gathered}
 0,\qquad\operatorname{diag}(0,0,1),\qquad
 \operatorname{diag}(0,1,z),\quad z\ne0,1,\\
 J_2(0)\oplus[0],\qquad J_2(0)\oplus[1],\qquad J_3(0).
\end{gathered}
\end{equation*}
Here $0$ represents the scalar class, which gives only trivial charges. For 
a generator with three distinct eigenvalues, the matrix is diagonalizable to $\operatorname{diag}(\lambda_1, \lambda_2, \lambda_3)$. By shifting the spectrum by $-\lambda_1 I$ and rescaling by a factor of $(\lambda_2 - \lambda_1)^{-1}$, we reduce it to the following canonical form
\begin{equation*}
    w \sim \operatorname{diag}\left(0, 1, \frac{\lambda_3 - \lambda_1}{\lambda_2 - \lambda_1}\right) \equiv \operatorname{diag}(0, 1, z),
\end{equation*}
where $z\in \mathbb{C} \setminus \{0, 1\}$, to ensure the three eigenvalues remain distinct. Geometrically, the parameter $z$ represents a projective cross-ratio of the three eigenvalues together with a distinguished reference point at infinity, which is fixed under affine spectral transformations. Choosing $(x_1,x_2,x_3,x_4)=(\lambda_3,\lambda_2,\lambda_1,\infty)$, the cross-ratio reduces to
\begin{equation*}
(\lambda_3,\lambda_2;\lambda_1,\infty) =\lim_{x_4\to\infty}
\frac{(\lambda_3-\lambda_1)(\lambda_2-x_4)}
     {(\lambda_3-x_4)(\lambda_2-\lambda_1)}
=\frac{\lambda_3-\lambda_1}{\lambda_2-\lambda_1} =z.
\end{equation*}
Since the eigenvalues are unordered, permutations generate the \textit{anharmonic}\footnote{In classical projective geometry, when the cross-ratio equals $-1, 2,$ or $\frac{1}{2}$, the four points form a harmonic bundle (where the 6 group elements collapse to just 3 distinct values). For any generic value of $z$, the six values are distinct, and the point set is termed \textit{anharmonic}.}\cite{coxeter2003projective} equivalences
\begin{equation*}
 z\sim1-z\sim z^{-1}\sim(1-z)^{-1} \sim\frac{z}{z-1}\sim\frac{z-1}{z}.
\end{equation*}
Consequently, two diagonal forms $\operatorname{diag}(0,1,z)$ and $\operatorname{diag}(0,1,z')$ are equivalent precisely when $z$ and $z'$ belong to the same anharmonic orbit.

For Hermitian generators, the classification can be organized
by their eigenvalues. A unitary change of basis diagonalizes
$w$, while replacing $w$ by $\alpha w+\beta I$, with real
$\alpha\ne0$ and $\beta$, leaves both bond commutants unchanged.
We may therefore shift and rescale the eigenvalues to simplify
their description. For a qubit, every non-scalar Hermitian
generator has two distinct eigenvalues and is consequently
equivalent to $S^z$. For spin 1, a Hermitian generator with
three distinct eigenvalues can instead be normalized to
$\operatorname{diag}(0,t,1)$, where $0<t<1$. The parameter $t$
measures the position of the middle eigenvalue between the
smallest and largest. Reversing the ordering identifies $t$
with $1-t$, and the equally spaced spectrum of $S^z$
corresponds to $t=\tfrac12$. Generators with repeated
eigenvalues form separate classes.

Once $w$ is diagonal, with eigenvalues $\omega_m$, the product
state $|m,n\rangle$ has bond eigenvalue
$\omega_m+\omega_n$ in the sum case and $\omega_m-\omega_n$
in the difference case. An operator in the bond commutant
may mix precisely those product states having the same bond
eigenvalue. Thus the commutant is obtained by grouping equal
sums or differences and allowing an arbitrary matrix within
each group. For non-Hermitian generators that are not
diagonalizable, one must instead solve the linear commutation
equations using the Jordan form of the bond operator.

While we systematically generate Hamiltonian families using these methods, there could still be some redundant equivalences among them. Eliminating these redundancies is the final step for such a classification which we leave for the future.


\section*{Acknowledgements}

The authors acknowledge the use of GPT-6 Astra to assist with manuscript preparation and to cross-check mathematical expressions and derivations. All AI-assisted content was reviewed by the authors, who take full responsibility for the final manuscript and the validity of its results.

\appendix
\section{Scalene realization of product-bond commutants}
\label{app:product-bonds}
Let $w$ be an invertible one-site operator and choose a one-site operator $r$ with $r^2=I$; the choice $r=I$ is sufficient. Define
\begin{equation*}
 A_{12}=w_1^{-1}r_2,\qquad
 B_{13}(u)=P_{13}\check B_{13}(u),\qquad
 C_{23}=r_2w_3.
\end{equation*}
Then the scalene equation is equivalent to the product-bond commutant condition,
\begin{equation*}
 A_{12}B_{13}(u)C_{23}=C_{23}B_{13}(u)A_{12}
 \quad\Longleftrightarrow\quad
 [\check B_{13}(u),w_1w_3]=0.
\end{equation*}
Indeed, moving $P_{13}$ to the left and using $r_2^2=I$ gives
\begin{align*}
 A_{12}B_{13}C_{23}-C_{23}B_{13}A_{12}
 &=P_{13}\bigl(w_3^{-1}\check B_{13}w_3
                  -w_1\check B_{13}w_1^{-1}\bigr)\\
 &=P_{13}w_3^{-1}[\check B_{13},w_1w_3]w_1^{-1}.
\end{align*}
Invertibility of the outside factors proves the equivalence. Any density satisfying $[h_{j,j+1},w_jw_{j+1}]=0$ admits the regular choice $\check B_{j,j+1}(u)=I+u h_{j,j+1}$. The densities and regular families may be chosen independently on each bond, using the assignment in \eqref{eq:inhomogeneous-derivative}. This factorized realization needs no additional algebraic relation on $w$, except invertibility. Independent nonzero overall scalar factors multiplying $A$ and $C$ also preserve the equation.

With the auxiliary labels of Section~\ref{sec:local-scalene}, its monodromy and transfer matrix are
\begin{equation*}
 T_b^{(C)}=r_b^L G,\qquad
 G=\prod_{j=1}^Lw_j,\qquad
 \tau_C=\operatorname{tr}_b(r_b^L)G.
\end{equation*}
In particular, $r=I$ gives $\tau_C=d_bG$ for every $L$. The mixed $RTT$ relation reduces to
\begin{equation*}
 w_a^{-1}T_a^{(B)}(u)G=G T_a^{(B)}(u)w_a^{-1},
 \qquad [T_a^{(B)}(u),w_aG]=0.
\end{equation*}
Conjugating by $w_a$ and taking the auxiliary trace yields $[\tau_B(u),G]=0$. Regularity therefore gives $[H,G]=0$ for $H=\tau_B(0)^{-1}\tau_B'(0)$. Equivalently, each local commutator vanishes directly because all other factors of $G$ act outside its bond. Unlike the additive charge cases, there are no restrictions on the length of the lattice here. For unitary $w$, this is the symmetry $GHG^{-1}=H$; for general invertible $w$, it is an algebraic conservation law. This appendix establishes the local scalene equivalence and its global consequence, without assuming ordinary Yang-Baxter integrability. Singular $w$ and the necessity of a local product-bond relation for a prescribed global product symmetry, require separate treatment and are left open.

\end{document}